\documentclass[11pt]{article}

\usepackage{graphicx}
\usepackage{cite}
\usepackage{color}
\usepackage{amsfonts}
\usepackage{amsthm}
\usepackage{url}
\usepackage{amsmath}
\usepackage{amssymb}
\usepackage{enumerate}
\usepackage{subfig}
\usepackage{fixltx2e}
\usepackage{dsfont}
\usepackage{pseudocode}
\usepackage{mathrsfs}

\newtheorem{theorem}{Theorem}
\newtheorem{lemma}{Lemma}

\newcommand{\EE}{\mathbb{E}}
\newcommand{\R}{\mathbb{R}}

\newcommand{\mA}{\mathsf{A}}

\renewcommand{\to}{\longrightarrow}
\usepackage{eufrak}
\graphicspath{{figs/}}
\usepackage[numbers]{natbib}

\usepackage{fullpage}
\title{Transportation Proof of an inequality by  Anantharam, Jog and Nair}
\author{Thomas~A.~Courtade\\University of California, Berkeley}
\date{January 30, 2019}
\begin{document}

\maketitle 

\begin{abstract}
Anantharam, Jog and Nair recently put forth an entropic inequality which simultaneously generalizes  the Shannon-Stam entropy power inequality and the Brascamp-Lieb inequality in entropic form. We give a brief proof of their result based on optimal transport. 
\end{abstract}

\section{Introduction}
Let $(c_1, \dots, c_k)$ and $(d_1, \dots, d_m)$ be nonnegative numbers, and let $(\mA_1, \dots, \mA_m)$ be a collection of surjective linear transformations identified as matrices, satisfying $\mA_j : \R^n \to \R^{n_j}$ for $j=1,\dots, m$.  Let $\mathcal{S}^{+}(\R^n)$ denote the set of $n\times n$ real symmetric, positive definite matrices, and define 
$$
M_g := \sup_{B_1, \dots B_k: B_i \in \mathcal{S}^{+}(\R^{r_i})}  \frac{1}{2} \sum_{i=1}^k c_i   \log  \det( B_i ) -   \frac{1}{2}\sum_{j=1}^m d_j   \log  \det( \mA_j B \mA_j^T ),
$$
where $B:=\operatorname{diag}(B_1, \dots, B_k)$.   For a random vector $X$ in $\R^n$ with density $f$ with respect to Lebesgue measure, we define the Shannon (differential) entropy according to
$$
h(X) := -\int_{\R^n} f(x) \log f(x) dx,
$$
and say that the entropy exists if the defining integral exists  in the Lebesgue sense and is finite. 

Anantharam, Jog and Nair recently established the following result:
\begin{theorem}[{\cite[Theorem 3]{anantharam2019unifying}}]\label{thm:BLEPI}
Let $X$ be a random vector in $\R^n$ that  can be partitioned into $k$ mutually independent components $X = (X_1, \dots , X_k)$, where each $X_i$ is a random vector in $\R^{r_i}$, and $\sum_{i=1}^k r_i  = n$. If $X$ has finite entropy and second moments, then letting the above notation prevail, 
$$
\sum_{i=1}^k c_i h(X_i) - \sum_{j=1}^m d_j h(\mA_j X) \leq M_g.
$$
\end{theorem}
As discussed in detail in \cite{anantharam2019unifying}, this result contains as special cases both the Shannon-Stam entropy power inequality \cite{bib:Shannon48, stam1959some} (in Lieb's form \cite{lieb1978proof}), and the Brascamp-Lieb inequality  (in entropic form, due to Carlen and Cordero-Erasquin \cite{carlen2009subadditivity}).  

Anantharam \emph{et al.}'s proof of Theorem \ref{thm:BLEPI} is based on a doubling  argument applied to information measures, following the scheme developed in \cite{GengNair} by Geng and Nair.  This doubling-trick for proving Gaussian optimality  goes back at least to  Lieb's original proof of the Brascamp-Lieb inequality \cite{lieb1990gaussian}, but the Geng-Nair interpretation in the context of information measures has enjoyed recent popularity in information theory (e.g., \cite{courtade2018strong, sula2018sum, liu2018forward}). 
The contribution in the present note is to give a brief proof of Theorem \ref{thm:BLEPI} based on optimal transport.  It has the advantage of being considerably shorter than the doubling proof in \cite{anantharam2019unifying}.   Interestingly, the proof here also seems  to be simpler than Barthe's transport proof of the Brascamp-Lieb inequality \cite{barthe1998reverse}.  However, Barthe's argument and the proof contained herein are  not truly comparable on account of  the following caveats:  (i) Theorem \ref{thm:BLEPI} implies the entropic form of the Brascamp-Lieb inequality, so some work is required to recover the functional form; and (ii) Barthe's argument simultaneously establishes a reverse form of the Brascamp-Lieb inequality (i.e., Barthe's inequality), and further gives a precise relationship between best constants in the forward and reverse inequalities. 

\section{Proof of Theorem \ref{thm:BLEPI}}
The key lemma is the following change-of-variables estimate, inspired by Rioul and Zamir's recent proof \cite{rioul2019equality} of  the Zamir-Feder entropy power inequality \cite{zamir1993generalization} (which also follows from Theorem \ref{thm:BLEPI}, as noted in \cite{anantharam2019unifying}).    We remark that other applications of optimal transport to entropy power inequalities can be found in \cite{rioul2017yet, rioul2017optimal, courtade2018quantitative}.  Readers are referred to \cite{villani2003topics} for background on optimal transport. 
\begin{lemma}\label{lem:ChgOfVar}
Let $\tilde{Z} \sim N(0,I)$ be a standard normal random variable in $\R^n$, and let $\mA :\R^n \to \R^m$ be a surjective linear map. 
Let $X$ be a random vector in $\R^n$, and let $T:\R^n\to \R^n$ be the Brenier map sending $\tilde{Z}$ to $X$.  If $T$ is differentiable with pointwise positive definite Jacobian $\nabla T$, then 
$$
h(\mA X) \geq h(Z) + \frac{1}{2}\EE\log \det (\mA  ( \nabla T(\tilde Z) )^2   \mA^T ),
$$
where $Z$ is standard normal on $\R^m$. 
\end{lemma}

\begin{proof}
Consider the QR decomposition of $\mA^T = QR =  [Q_1,Q_2][R_1^T, 0]^T$, where  $Q$ is an orthogonal $n\times n$ matrix,  and $R_1$ is an upper triangular $m\times m$ matrix, with positive entries on the diagonal. 
Let $Z'$ be standard normal on $\R^{n-m}$, independent of $Z$, and note that $\tilde{Z} = Q_1 Z + Q_2 Z'$ is a valid coupling.  Now, for fixed $z'$, the map $z\in \R^m  \to \mA  T( Q_1 z + Q_2 z')\in \R^m $  is invertible and differentiable.  Differentiability follows from our assumption on $T$, and invertibility follows by writing $T = \nabla \varphi$ for strictly convex $\varphi$ (Brenier's theorem with the positivity assumption), and noting that the map  $z  \to Q_1^T  \nabla \varphi( Q_1 z + Q_2 z')$ is the gradient of the strictly convex function $z  \to   \varphi( Q_1 z + Q_2 z')$, and is therefore invertible.  So, we have
\begin{align}
h(\mA X) = h(\mA  T(\tilde Z ))
&=h(\mA  T( Q_1 Z + Q_2 Z' ) ) \label{eqTransport}\\
&\geq h(\mA  T( Q_1 Z + Q_2 Z') | Z' ) \label{eq:condRedEntropy}\\
&=h(Z) + \EE\log  \det( Q_1^T \nabla T(\tilde Z)   Q_1) + \log \det R_1 \label{eq:CoV1}\\
&=h(Z) + \frac{1}{2}\EE\log \left[( \det( Q_1^T \nabla T(\tilde Z)  Q_1) )^2\right]+ \log \det R_1 \notag\\
&=h(Z) + \frac{1}{2}\EE\log  \det( Q_1^T  (\nabla T(\tilde Z))^2  Q_1) + \log \det R_1 \label{specSquared}\\
&=h(Z) + \frac{1}{2}\EE\log  \det( \mA   (\nabla T(\tilde Z))^2  \mA^T ).\notag
\end{align}
Above, \eqref{eqTransport} follows since $X = T(\tilde{Z})$ in distribution; \eqref{eq:condRedEntropy} follows from the fact that conditioning reduces entropy; \eqref{eq:CoV1} is the change of variables formula for entropy; and \eqref{specSquared} follows since the squared spectrum of $Q_1^T \nabla T(\tilde z)  Q_1$ is equal to the spectrum of $Q_1^T (\nabla T(\tilde z))^2  Q_1$ for each $\tilde{z}$ by symmetry of $\nabla T$ (an easy exercise, e.g., seen by diagonalizing $\nabla T(\tilde z)$). 
\end{proof}

Now, we begin the proof of Theorem \ref{thm:BLEPI}.    Without loss of generality, we may assume the density of each $X_i$ is smooth, bounded and strictly positive.  Indeed, if this is not the case, then we first regularize the density of $X$ via convolution with a Gaussian density.  The general claim then follows by continuity in the limit of vanishing regularization, which is valid provided entropies and second moments are finite (e.g.,  \cite[Lemma 1.2]{carlen1991entropy}).    

By dimensional analysis, a necessary condition for $M_g<\infty$ is that  $\sum_{i=1}^k c_i r_i = \sum_{j=1}^m d_j n_j$.  So, we make this assumption henceforth.  Now, let $Z = (Z_1, \dots, Z_k)$ be independent, standard normal random vectors with $Z_i\in \R^{r_i}$, and let $T_i:\R^{r_i}\to \R^{r_i}$ be the Brenier map sending $Z_i$ to $X_i$.    Define $T = (T_1, \dots, T_k)$, which is the Brenier map transporting $Z$ to $X$ by the independence assumption.  We remark that each $T_i$ is differentiable, with $\nabla T_i$ being   pointwise symmetric and positive definite.  This follows from  Brenier's Theorem \cite{brenier1991polar} and regularity estimates for the Monge-Amp\`ere equation under our assumption that the densities of the $X_i$'s are smooth with full support  \cite[Remark 4.15]{villani2003topics}. 

So, by Lemma \ref{lem:ChgOfVar}, we have
\begin{align}
\sum_{j=1}^m d_j h(\mA_j X) &\geq \sum_{j=1}^m d_j h(Z'_j) + \frac{1}{2}\sum_{j=1}^m d_j  \EE\log  \det( \mA_j   (\nabla T(Z))^2  \mA_j^T ),
\end{align}
where $Z'_j$ is standard normal on $\R^{n_j}$.  
By the change of variables formula,   $h(X_i) =  h(Z_i) +    \EE\log  \det(   \nabla T_i(Z_i) )$ for each $i=1, \dots, k$, so summing terms gives
$$
\sum_{i=1}^k c_i h(X_i) = \sum_{i=1}^k c_i h(Z_i) + \frac{1}{2} \sum_{i=1}^k c_i  \EE\log  \det(   (\nabla T_i(Z_i))^2   ).
$$
By the relation $\sum_{i=1}^k c_i r_i = \sum_{j=1}^m d_j n_j$, we have $\sum_{i=1}^k c_i h(Z_i) =  \sum_{j=1}^m d_j h(Z'_j)$.  Hence, on combining the above estimates, we have
\begin{align*}
\sum_{i=1}^k c_i h(X_i) - \sum_{j=1}^m d_j h(\mA_j X) &\leq \frac{1}{2} \sum_{i=1}^k c_i  \EE\log  \det(   (\nabla T_i(Z_i))^2    ) -   \frac{1}{2}\sum_{j=1}^m d_j  \EE\log  \det( \mA_j   (\nabla T(Z))^2   \mA_j^T )\\
&\leq M_g,
\end{align*}
where the last line follows by definition of $M_g$ (applied pointwise inside the expectation).

\bibliographystyle{unsrt}
\bibliography{BLEPI}

\end{document}